\documentclass[a4paper]{article}
\usepackage[margin=25mm]{geometry}
\usepackage{amsfonts}
\usepackage{amssymb}
\usepackage{graphicx}
\usepackage{algorithm}  
\usepackage{algorithmicx}  
\usepackage{algpseudocode}  

\usepackage{amsmath,amsthm,amssymb}

\usepackage{graphicx}
\usepackage{float}

\newtheorem{theorem}{Theorem}[section]
\newtheorem{remark}{Remark}[section]
\newtheorem{lemma}{Lemma}[section]
\newtheorem{definition}{Definition}[section]
\newtheorem{Conjecture}{Conjecture}[section]
\newtheorem{corollary}{Corollary}[section]

\floatname{algorithm}{algorithm}

\usepackage{verbatim}
\immediate\write18{texcount -tex -sum  \jobname.tex > \jobname.wordcount.tex}

\providecommand{\keywords}[1]
{
  \small	
  \textbf{\textit{Keywords---}} #1
}

\title{An Adaptive Grid Algorithm for Computing the Homology Group of Semialgebraic Set
}
\author{Jiadong Han $^{1}$
      \\
        \small $^{1}$Universit\'e Paris Sud \\
         \\
}
\date{} % Comment this line to show today's date
\usepackage[style=alphabetic,citestyle=alphabetic,backend=biber,backref=true]{biblatex}
\begin{document}
\maketitle

\begin{abstract}
Looking for an eﬃcient algorithm for the computation of the homology groups of an algebraic set or even a semi-algebraic set is an important problem in the eﬀective real algebraic geometry. Recently, Peter B$\ddot{u}$rgisser, Felipe Cucker and Pierre Lairez wrote a paper \cite{1}, which made a step forward by giving an algorithm of weak exponential time. However, the algorithm has not yet became practical. In our thesis, I will introduce our work on an improvement of this algorithm using an adaptive grid algorithm on the unit sphere.
\end{abstract} \hspace{10pt}

%TC:ignore

\setcounter{section}{0}
\begin{section}{Introduction}
At first, let us introduce some basic deﬁnitions in this thesis.

A basic semialgebraic set is a subset of a Euclidean space $R^n$ given by a system of equalities and inequalities of the form
$$f1(x) = ... = fq(x) = 0~~ and ~~g_1(x) \succ 0,...,g_s(x) \succ 0 (\ast)$$
where $F = (f_1,...,f_q)$ and $G = (g_1,...,g_s)$ are tuples of polynomials with real coeﬃcients and the expression $g(x) \succ 0$ stands for either $g(x) \geq 0$ or $g(x) > 0$ (we use the notation to emphasize the fact, which become clear in \cite{1}\textbf{4.1.4}  where the author’s main results do not depend on whether the inequalities in $(\ast)$ are strict.) Let $W(F,G)$ denote the solution set of the semialgebraic system $(\ast)$. And let $S(F,G)$, called the spherical semi algebraic set, denote the solution set of the semialgebraic system $(\ast)$ in the unit sphere. 

For a vector $\mathbf{d} = (d_1,...,d_{q+s})$ of $q + s$ positive integers, we denote by $\mathcal{P}_d$ (or $\mathcal{P}_d[q + s]$ to emphasize the number of components) the linear space of the $(q + s)$-tuples of real polynomials in $n$ variables of degree $d_1,...,d_{q+s}$, respectively. Similarly, for a vector
$\mathbf{d} = (d_1,...,d_{q+s})$ of $q +s$ positive integers, we denote by
$\mathcal{H}_d$ (or $\mathcal{H}_d[q +s]$ to emphasize the number of components)
the linear space of the $(q + s)$-tuples of real homogeneous polynomials in $n$
variables of degree $d_1,...,d_{q+s}$, respectively.

Let $D$ denote the maximum of the $d_i$. We will assume that $D \geq 2$ because a set deﬁned by degree $1$ polynomials is convex and its homology is trivial. Let $N$ denote the dimension of $\mathcal{P}_d$, that is, $N = \sum^{q+s}_{i=1}\binom{n+d_i}{n}$.

This is the size of the semialgebraic system $(\ast)$, as it is the number of real coeﬃcients necessary to determine it.

We will endow $\mathcal{P}_d$ with the Weyl inner product and induced norm in $\S3.1$. The Weyl inner product is a dot product with respect to a specially weighted monominal basis. In particular, it is invariant under orthogonal transformations of the homogeneous variables $(X_0,...,X_n)$. That is, for any orthogonal transformation $u : R^n → R^n$ and any $f \in \mathcal{H}_d[q]$, we have $||F|| = ||F \circ u||$. In all of what follows, all occurences of normes in spaces $\mathcal{H}_d[q]$ refer to the norm induced by the Weyl inner product. For a point $x \in R^{n+1}$ and a system $F \in \mathcal{H}_d[q]$, let $DF(x)$ denote the derivate of $F$ at $x$, which is a linear map $R^{n+1} \rightarrow R^q$. We also deﬁne the diagonal normalization matrix
$$\left\{
  \begin{matrix}
   \sqrt{d_1} &  &  \\
   \empty     & ... & \empty \\
     &   & \sqrt{d_q}
  \end{matrix} \right\} $$
  
We will use these to define some $D -Lipschitz$ continuous condition functions in $\ref{Theorem 3.7.}$. From these condition numbers, we can get \cite{1} \textbf{Theorem 1.1} to compute homology group of a semi-algebraic set. 

Now we will outline how the authors in [1] get the above theorem and think about how to get an improvement of this theorem.In \cite{2}, Niyogi, Smale and Weinberger give an answer to the following question: giving a compact submanifold $S \subset E$, a ﬁnite set $\mathcal{X} \subset E$ and $\epsilon > 0$, how to ensure that $S$ is a deformation retract of $\mathcal{U}(\mathcal{X},\epsilon)$ (see \S3.1 ). And \cite{1} gives an extension in the \textbf{Theorem 2.8. } Based on this theorem, the paper \cite{1} gives a covering algorithm to generate a ﬁnite point set on the sphere whose neighbourhood is homotopical to the semialgebraic set S(F,G) in \textbf{Theorem 1.8.}  Then the authors give an algorithm on the calculation of the homotopy group. The key point is that the homotopy equivalence implies the isomorphism of homology group. 

As a result, one can get the homology group by computing the nerve $\mathcal{N}$ of the covering $\{B(x,\epsilon)|x \in \mathcal{X}\}$ (this is the simplicial complex whose elements are the subsets $\mathcal{N}$ of $\mathcal{X}$ such that $\cap_{x\in\mathcal{N}}B(x,\epsilon)$ is not empty) and computing it's homology group by the Nerve Theorem (Corollary 4G.3 of \cite{4}). The process is described in detail in \cite{5} \S4 where the proof of the following result can be found. In \cite{6} and \cite{7} there are improved algorithms for computing the nerve of a covering. 

\begin{theorem}\label{Theorem 1.9.}
(\cite{6} and \cite{7} ) Given a ﬁnite set $\mathcal{X} \subset R^{n+1}$ and a positive real number $\epsilon$, one can compute the homology of $\cup_{x\in\mathcal{X}}B(x,\epsilon)$ with $\mathcal{X}^{O(n)}$ operations.

\end{theorem}

Now the authors give the proof of the Theorem 1.1 in \cite{1}.

However, in practice, the above algorithms is slow. One reason is that we only allow the division of the same radius in the cube, which may generate many overqualiﬁed points. As a result, we will introduce a more adaptive grid algorithm in chapter 2. Based on this grid algorithm, we will prove a local version in the the Niyogi-Smale-Weinberger theorem, which allows us to apply our grid algorithm on the computation of homology groups. 

Let me introduce our strategy in this thesis.

In the chapter 2, we give an adaptive grid algorithm on the unit shpere, which is a basis for our future design. This algorithm is a generalized idea of ﬁnding a minimum of a Lipchitz continuous function on an interval. In the intuition, Finding a minimum of a Lipchitz continuous function f on an interval $[0,1]$ will lead to the complexity of $O(\max_{x\in[0,1]}\frac{1}{  f(x)})$ in \ref{Theorem 2.1}. However, we give an algorithm of the complexity $O(\int_{[0,1]}\frac{1}{f(x)})dx)$ in \ref{Theorem 2.3}. Then ze generaliwe it to the sphere in \ref{Theorem 2.5}.

From this algorithm, we can design an adaptive covering algorithm to generate a ﬁnite point set on the sphere whose neighbourhood is homotopical to the semialgebraic set $S(F,G)$.

For the goal of our thesis, we need a theorem to judge whether a compact set in Sn is homotopical to $S(F,G)$. So we prove an extension of the Niyogi-Smale-Weinberger theorem in a unit sphere in \ref{Corollary 3.24.}.To prove this theorem, we at ﬁrst give a local version of the Niyogi-Smale-Weinberger theorem in \ref{Theorem 3.15}. Then we prove it by generalizing the the proof on \cite{1}.

Based on the adaptive grid algorithm and the Niyogi-Smale-Weinberger theorem, we can design the algorithm in \ref{Theorem 3.25.}.Notice \ref{Theorem 1.9.} dose not depend on the choice of $\epsilon$. As a result, it is quite reasonable to have the following conjecture. 
\begin{Conjecture} \label{Conjecture 1.17.} Given a ﬁnite set $\mathcal{X} \subset R^{n+1}$ and a positive $A-Lipschitz$ continuous function $\epsilon(x)(A > 0)$, one can compute the homology of $\cap_{x\in \mathcal{X}}B(x,\epsilon(x))$ with $\mathcal{X}^{O(n)}$ operations.
\end{Conjecture}

Now we can give a more adaptive homology group calculation algorithm and give a conjectured complexity in \ref{Theorem 3.28} which reduced the complexity greatly.
\end{section}

\begin{section}{An adaptive grid algorithm}
In this chapter, we will introduce an adaptive grid algorithm to obtain an optimized algorithm for the computation of homology group in a semi-algebraic set. At ﬁrst, we will explain the potential grid algorithm in the $Proposition 5.1$ of \cite{1}. Then we will give an adaptive version in interval, or dimension 1. Finally, we will give a much more general version in the unit sphere $\mathbb{S}^n$. As a most useful special case in our work, we will give an adaptive grid algorithm in a unit sphere.
We assume that $f : M \rightarrow (0,1]$ an $A-Lipschitz$ continuous function with $A ≥ 1$ in this section. And suppose that the complexity of the following algorithms are only the total number of nodes in the corresponding decision tree.

\subsection{An nonadaptive grid algorithm}
To introduce nonadaptive and adaptive grid algorithms, we consider the problem of computing the minimum of a 1-Lipschitz function f : [a,b] → (0,1].

At ﬁrst I will give an nonadaptive algorithm.

\begin{theorem}\label{Theorem 2.1}
 Given a small $\epsilon > 0$ and a $1-Lipschitz$ function $f : [0,1] \rightarrow (0,1]$, we have the following algorithm to ﬁnd a minimum of $f$. Concretely, there is an interval $[c,d] \subset [0,1]$ such that $| \frac{\min_{x\in [0,1]} f(x)} {\min_{x\in\{c,d\}} f(x) }-1| \leq \epsilon$. This algorithm performs $O(\frac{1}{\epsilon} \max_{x\in [0,1]} \frac{1}{ f(x)})$ evaluations on $f$.

\end{theorem}

\begin{algorithm}  
        \caption{}  
        \begin{algorithmic}[1] 
            \Require  a small $\epsilon > 0$ and a 1-Lipschitz function $f : [0,1] \rightarrow (0,1]$
            \Ensure   a interval [c,d] such that $|\frac{\min_{x\in [0,1]} f(x)} {\min_{x\in\{c,d\}} f(x) }-1| \leq \epsilon$
            \State $G \gets\{0,1\}$
            \State $k \gets 0 $
            \While {$r \geq \epsilon f(x)$ for a $x \in G$ }
              \State $k \gets k + 1$ and $G \gets {i2^{−k}|0 \leq i \leq 2^k}$
            \EndWhile 
            \State \Return { $\min_{x\in G} f(x)$}

        \end{algorithmic}  
    \end{algorithm}  

\begin{remark}
This algorithm is nonadaptive since the grid generation process. We can see each grid generation process as a division of generated intervals together.
\end{remark}

\begin{proof}
At ﬁrst, we prove the correctness of this algorithm. Let $m = \min_ {x∈{0,1}}f(x)$. Notice
that if $r \leq m\epsilon$, we see that 
$|\frac{\inf_{[a,b]}(f)}{ m} −1|\leq\frac{r}{m} \leq \epsilon$. Take an interval $[a,b] \subset [0,1]$. By the recursion process, we see that the algorithm is correct when $|a−b| = r$ implies the algorithm is correct when $|a−b| = 2r$. As a result, for the interval $[0,1]$, we get the algorithm is correct from this induction process.
Then, we will prove that this algorithm will terminate. Notice that if $r ≤ m\epsilon$, the algorithm will satisfy the stop condition. Take an interval $[a,b] \subset [0,1]$. By the recursion, we see that the algorithm will terminate if $|a−b| = r $implies that the algorithm will ﬁnish if $|a−b| = 2r$. As a result, for the interval $[0,1]$, we get the algorithm will terminate from this induction process.
Now let us give a complexity analysis. Notice the $n^{th}$ iteration generates $2n$ intervals. Notice there must be a point in $[0,1]$ where $f(x)$ takes the minimum. As a result, suppose the maximum iteration time is $k$, we have $\frac{1}{ 2^k} ≤ \min_{x\in[0,1]} f(x)\epsilon$, or$ k \leq \log(\min_{x\in[0,1]} f(x)\epsilon) + 1$. As a result, when we assume an eﬃcient computation in the evaluation, the complexity is bounded by
$$\sum^{\lceil\max_{x\in [0,1]}(-\log(f(x)\epsilon)+1)\rceil}_{i=0}2^i =
2^{\lceil\max_{x\in [0,1]}(-\log(f(x)\epsilon)+2)\rceil } -1 \leq \mathcal{O}
(\frac{1}{\epsilon}\max_{x\in [0,1]}\frac{1}{f(x)}) $$
\end{proof}

\subsection{ An adaptive grid algorithm in an interval }

Notice that the complexity of this nonadaptive algorithm depends on the maximum of $\frac{1}{ f(x)}$ because there are many useless computations. For example, when we get an interval which satisﬁes the condition $r < \epsilon \min\{f(a),f(b)\}$, we still divide this interval. As a result, a precise partition of the generated intervals should be added to the design.
Now we have the following algorithm to calculate the minimum of a strictly positive 1-Lipschitz function f on $[a,b]$ with the complexity $\mathcal{O}(\frac{1}{\epsilon}\int_0^1\frac{1}{f(x)}).$

\begin{theorem}\label{Theorem 2.3}
  Given a closed interval $[a,b] \subset [0,1]$, a small $\epsilon > 0$ and a $1-Lipschitz$ function $f : [a,b] \rightarrow (0,1]$, we have the following algorithm to to ﬁnd a minimum of $f$. Concretely, there is an interval $[c,d] \subset [0,1]$ such that $ | \frac{\inf_{[a,b]}(f)}{min_{x\in\{c,d\}}f(x)}  −1| \leq \epsilon$. This algorithm performs $\mathcal{O}(\frac{1}{\epsilon}\int_0^1\frac{1}{f(x)})$ evaluations on $f$. 
 
\end{theorem}

\begin{algorithm}  
        \caption{}  
        \begin{algorithmic}[1] 
            \Require  a small $\epsilon > 0$ and a 1-Lipschitz function $f : [0,1] \rightarrow (0,1]$
            \Ensure   a interval [c,d] such that $|\frac{\min_{x\in [0,1]} f(x)} {\min_{x\in\{c,d\}} f(x) }-1| \leq \epsilon$
            \State $L_{todo} \gets \{[a,b]\}$
            \State $L_{final} \gets \emptyset  $
            \While { $L_{todo}$ is not empty }
              \State $L_{next} = \emptyset$
              \For{$[c,d] \in L_{todo}$ }
                    \If{ $\min\{f(c),f(d)\}\epsilon \geq|c−d|$ }
                        \State $L_{final} \gets L_{final} \cup [c,d].$
                    \Else
                        \State $L_{next} \gets L_{next} \cup \{[c, \frac{c+d }{2} ],[\frac{c+d }{2} ,d]\}.$
                    \EndIf               
              \EndFor
            \State $L_{todo} \gets L_{next}$   
            \EndWhile 
            \State \Return { $L_{final}$}

        \end{algorithmic}  
    \end{algorithm}  

\begin{proof}
At ﬁrst, we prove the correctness of this algorithm. Given an interval $[a,b]$, a strictly positive $1-Lipschitz$ function $f$ and a very small positive integer $\epsilon > 0$. Let $m = \min _{x\in {a,b}} f(x)$. Notice that if $r \leq m\epsilon$, we see that $| \frac{\inf_{[a,b]}(f)}{m } −1| ≤ r m ≤ \epsilon.$ By the recursion process, we see that the algorithm is correct if $|a−b| = r$ implies the algorithm is correct if $|a−b| = 2r$. As a result, for any interval $[a,b]$, we get the algorithm is correct by the induction process.
Then, we prove that this algorithm can terminate . Notice that if $r ≤ m\epsilon $, the algorithm will satisfy the stop condition. By the recursion, we see that the algorithm will terminate if $|a−b| = r $implies that the algorithm will terminate if $|a−b| = 2r$. As a result, for any interval $[a,b]$, we get the algorithm will terminate by the induction process.
At last, let us analyze the complexity of this algorithm.
Notice that we can visualize the process of the algorithm as a binary tree. The nodes of the tree are input intervals in each iteration of the algorithm, the $n^th$ layer is the process after $n−1$ iterations and the leaves are the ﬁnal iteration in each branches. In addition,
the complexity is less than the double of cardinality of leaves For any $x$ in a leaf $[a,b]$, we have $$f(x)\epsilon = (f(x)−f(a) + f(a))\epsilon \leq |f(x)−f(a)|\epsilon  + f(a)\epsilon  \leq  \epsilon |x−a|+ 2|a−b|\leq \epsilon|b−a|+ 2|a−b| = (2 + \epsilon)|a−b|$$
and
$$|a−b|\leq f(a)\epsilon \leq|f(a)−f(x)|\epsilon + f(x)\epsilon  \leq |a−x|\epsilon + f(x)\epsilon ≤|a−b|\epsilon + f(x)\epsilon.$$
In summary, we get
$\frac{1}{2+\epsilon}
f(x) \leq |a−b|\leq
\frac{1}{1−\epsilon}
f(x)$.
By dividing by $f(x)\epsilon$ and taking the integral over $[a,b]$, we get
$\frac{1}{2 + \epsilon}  \leq
\frac{1}{\epsilon}\int_b^a\frac{1}{f(x) }
\leq
\frac{1}{1−\epsilon}$
.
Suppose there are $m$ leaves $\{[a_i,b_i]\}$. By the process of the recursion algorithm, we can see m as the operations times and the union of all leaves is $[0,1]$. After the m sums of the above inequality, we get
$$\frac{1-\epsilon}{\epsilon} \int_b^a
\frac{1}{f(x)} \leq m ≤
\leq \frac{2 + \epsilon}{\epsilon} \int_b^a
\frac{1}{f(x)} .$$
As a result, we get the complexity of the algorithm as $\mathcal{O}(\frac{1}{\epsilon}\int_0^1\frac{1}{f(x)}).$ 
\end{proof}

\begin{subsection}{An adaptive grid algorithm in a unit sphere
}
Inspired by $\mathbf{Theorem\ref{Theorem 2.3}}$, we want to ﬁnd an adaptive grid algorithm in a unit sphere. Actually, we can even design a similar algorithm in a certain Riemennian manifold.
However, we notice that the above algorithm analysis greatly depends on how to divide the interval. For instance, there is no overlaps in each partition, which allows us to take an integral on a small interval and then taking the sum. However, in general, it is not easy to show the control of overlaps. So we need to consider some additional conditions on the open cover.
At ﬁrst, we will give some useful notations.
Let $\mathcal{J}_k$ be a maximal set such that $\forall x,x′ \in \mathcal{J}_k, x \neq x′ \implies d_{\mathbb{S}^n}(x,x′) \leq 2^{−k} .$ Let $\mathcal{B}_k = \{B(x,2^{−k}) | x\in \mathcal{J}_k\}$.

\begin{lemma}
 Each $\mathcal{B}_k$ is an open cover of $\mathbb{S}^n$.
\end{lemma}

\begin{proof}
We can get this statement from a contradiction. Suppose there is a $x \in \mathbb{S}^n$ but is not covered by $\mathcal{B}_k$ . Then $\mathcal{J}_k \cup{x}$ satisﬁes the above condition that $\forall x,x′ \in  \mathcal{J}_k\cup\{x\}, x \neq x′ \implies d_{\mathbb{S}_n}(x,x′) \geq 2−k$ , which is contradicted to the maximal set of $ \mathcal{J}_k$.
\end{proof}

For a Ball $B$, we denote by $radius_B$ or $r_B$ its radius and $center_B$ or $c_B$ its center. For $B \in mathcal{B}_r$, let $divide(B) = \{B′ \in \mathcal{B}_{\frac{r}{2}} | B′\cap B \neq \emptyset \}$. We will also use the following notation: $union(x) = \cup_{A
\in X}A$. Notice that $ B\subset union(divide(B))$.
We deﬁne
$L_0 = B_0 L_{k+1} = union(\{divide(B)) | B \in \mathcal{L}_k and radius_B \geq f(center_B)\})$ for $k \geq 0$, and $\mathcal{M}_k = \cup^k_i=0{B \in \mathcal{L}i | radius_B < f(center_B)}$ for $k \leq 0$.
Now we design the algorithm and give a complexity analysis. Suppose An is the area of a $n−1$ dimension unit sphere.  Suppose $A_n$ is the area of a $n−1$ dimension unit sphere. 

\begin{theorem}\label{Theorem 2.5}
 Given a $A-Lipchitz (A \geq 0)$ continuous function $f : \mathbb{S}n \rightarrow (0,1]$, we have the following algorithm to generate a open cover $L_{final}$ of the unit sphere Sn such that $\forall B(c_B,r_{c_B}) \in L_{final}$ must satisfy $r_{c_B} \leq f(c_B)$. This algorithm performs less than $\frac{(16(2+A))^n}{ A_n log_e^2} \int_{\mathbb{S}^n}(\frac{1}{f} )ndV$ evaluations on $f$. Notice that $\mathcal{L}_final = \mathcal{M}_k$ where $k = \min\{k \in \mathbb{Z} |\mathcal{L}_k = \emptyset\}$ 
\end{theorem} 

\begin{algorithm}  
        \caption{ Cover(f) }  
        \begin{algorithmic}[1] 
            \Require  a small $\epsilon > 0$ and a 1-Lipschitz function $f : \mathbb{S}^n \rightarrow (0,1]$
            \Ensure    ﬁnite set $X = {B_{\mathbb{S}^n}(x,r_x)}$ of balls such that $\cup B_{\mathbb{S}^n}(x,r_x) = \mathbb{S}^n$ and $\forall x \in X, r_x \leq f(x) $
            \State $L_{todo} \gets \{B|B \in \mathcal{B}_0\}$
            \State $L_{final} \gets \emptyset  $
            \While { $L_{todo}$ is not empty }
              \State $L_{next} = \emptyset$
              \For{$B_{\mathbb{S}^n}(x,r_x) \in L_{todo}$ }
                    \If{ $ r < f(x) $ }
                        \State $L_{final} \gets L_{final} \cup \{B_{\mathbb{S}^n}(x,r_x)\}.$
                    \Else
                        \State $L_{next} \gets L_{next} \cup divide(B_{\mathbb{S}^n}(x,r_x)).$
                    \EndIf               
              \EndFor
            \State $L_{todo} \gets L_{next}$   
            \EndWhile 
            \State \Return { $L_{final}$}

        \end{algorithmic}  
    \end{algorithm}  
\begin{proof}

 Now we begin the proof of the correctness. From the judgement condition of the algorithm, we see that each output $B(c_B,r_{c_B})$ must satisfy $r_{c_B} \leq f(c_B)$. We will use the following lemma to get the output is an open cover of Sn. 
 
\begin{lemma}\label{Lemma 2.6}
  For all $k \geq 0$, $\mathcal{L}_k \subset \mathcal{B}_k$ and $union(\mathcal{L}_k)\cup union(\mathcal{M}_k) = \mathbb{S}^n$.
  \end{lemma}
\begin{proof}
$\mathcal{L}_k \subset \mathcal{B}_k$ is from the deﬁnition. We will prove $union(\mathcal{L}_k)∪union(\mathcal{M}_k) = \mathbb{S}^n$ by induction on the $k$. Suppose $k = 0$, we get the lemma from the deﬁnition of $\mathcal{B}_0$. Now suppose it is true for $k = m − 1$. We need to prove that it is true for $k = m$. Notice $\mathcal{M}_{k-1} \subset  \mathcal{M}_k$. For $\forall B \in  \mathcal{L}_{k-1}\supseteq \mathcal{M}_{k-1}$, we can have $divide(B) \subset \mathcal{L}_{k}$ and $B \subset union(divide(B))$. As a result, we get $union(\mathcal{L}_{k})\cup union(\mathcal{M}_{k}) \supseteq union(\mathcal{L}_{k-1}\setminus \mathcal{M}_{k-1})\cup  union(\mathcal{M}_{k-1}) = \mathbb{S}^n.
$
\end{proof} 
Notice from this lemma, the output of our algorithm is an open cover.
Now we will show that the algorithm will terminate. We need to note that if $r_{c_B} < \inf_{x\in M} f(x)$, the algorithm must terminate. In addition, it is easy to see that if the algorithm terminates when the ﬁnal input in the iteration process is $r$, the algorithm must terminate when the input in the iteration process has radius $2r$. As a result, the algorithm must terminate by the induction on the radius.
Before the later complexity analysis, we ﬁrst notice that on the requirement of the theorem, we can ﬁnd an upper bound of $divide(B)$.
\begin{lemma}\label{Lemma 2.7}
We have $\#divide(B) < 16^n$.
\end{lemma}
\begin{proof}

\end{proof}
 By considering the area of a disk with geodesic radius $r$ on $\mathbb{S}^n$ is the integral of $n−1$ dimensional spheres from radius $0$ to radius $sin(r)$, we get
$$Vol(B_{\mathbb{S}^n}(x,r)) = A_n\int^r_0
sin^{n−1}(t)dt.$$
Since $0 < t \leq 1$ from the design of our algorithm, we have $$\frac{A_n }{n} 2^{−n(k+1)} \leq Vol(B_{\mathbb{S}^n}(x,2^{−k})) \leq \frac{A_n }{n} 2^{−nk}.$$
Take $B \in \mathcal{B}_k$. From the deﬁnition of $divide(B)$ we have $B′ \in  divide(B)$ if and only if $center_{B′} \in (B_{\mathbb{S}^n}(c_{B}, 2^{−k} + 2^{−k-1})$ if and only if $B_{\mathbb{S}^n}(c_{B′},2^{−k−2}) \subset  B_{\mathbb{S}^n}(c_{B′},2^{−k} + 2^{−k−1} + 2^{−k−2}).$ Notice under our assumption, any two $B_{\mathbb{S}^n}(c_{B′},2^{−k−2})$ are disjoint. As a result, we have
$$\#divide(B) Vol(B_{\mathbb{S}^n}(c_{B′},2^{−k−2}) \leq Vol(B_{\mathbb{S}^n}(c_{B′},2^{−k} + 2^{−k−1} + 2^{−k−2}))$$.
Combine the two inequalities, we get $$\#divide(B) \leq\frac{\int^{2^{−k}+2^{−k−1}+2^{−k−2}}_0 sin^{n−1}(t)dt}{\int^{2{−k−2}}_0 sin^{n−1}(t)dt}   ≤\leq\frac{\int^{2^{−k}+2^{−k−1}+2^{−k−2}}_0 t^{n−1}dt}{\int^{2{−k−2}}_0 (\frac{t}{2})^{n−1}dt} ≤ 16^n$$.
Since $f(x) ≤ 1$ and $r_0 = 1$, we have
$$C \leq
\sum_{k=0}^{\infty}
\#\mathcal{L}_k =
\sum_{k=0}^{\infty} \sum_{B\in \mathcal{L}_k}\#divide(B) \leq  16^n
\sum_{k=0}^{\infty}\sum_{B\in \mathcal{L}_k}1
$$ $$\leq 16^n
\sum_{k=0}^{\infty}
\#\{B \in \mathcal{L}_k|r_B \geq f(c_B)\}\leq 16^n
\sum_{k=0}^{\infty}
\#\{B \in \mathcal{B}_k|r_B\geq f(c_B)\} (\ast )$$
To estimate the $\#\{B \in \mathcal{B}_k|r_B \geq f(c_B)\} $, we can use the similar trick as the estimate of the $\#divide(B)$. Finally, we get
$$(\ast) \leq  16^n
\sum_{k=0}^{\infty} \frac{Vol(\cup_{\{B\in \mathcal{B}_k}, 2−k≥f(c_B)\}B_{\mathbb{S}_n}(c_B,2^{−k−1}))}{ Vol(B_{\mathbb{S}_n}(c_B,2^{−k−1}))
}
(\ast \ast )$$
If $2^k \geq f(c_B)$, then $\forall x \in  B_{\mathbb{S}_n}(c_B,2^{−k−1})$, we have $$f(x) ≤ A2^{−k−1} + f(c_B) ≤ A2^{−k−1} + 2^{−k} = (1 + \frac{A}{2}
)2^{−k}.$$ 
So we get
$$(\ast \ast ) \leq 16^n
\sum_{k=0}^{\infty}
16^n \sum_{k=0}^{\infty} \frac{Vol(\{x \in \mathcal{S}^n|f(x) \leq (1 +\frac{A}{2})2^{−k}\})}{ Vol(B_{\mathbb{S}_n}(c_B,2^{−k−1}))} \leq \frac{16^n}{A_n} \sum_{k=0}^{\infty} \frac{Vol(\{x \in \mathcal{S}^n|f(x) \leq (1 +\frac{A}{2})2^{−k}\})}{ 2^{(−k−1)n}}$$ $$\leq \frac{n16^n}{A_n} \int_{0}^{\infty}2^{nt} Vol(\{x \in \mathcal{S}^n| (\frac{f}{2+A})^n \leq 2^{−nt}\})dt \leq \frac{n16^n}{A_n} \int_{0}^{\infty}2^{nt} \int_{\mathbb{S}^n} \mathbb{I}_{ (\frac{f}{2+A})^n \leq 2^{−nt}} dVdt ~(\ast \ast \ast )$$

From Fubini-Tonelli theorem, we have
$$(\ast \ast \ast) \leq
\frac{n16^n}{A_n}  \int_{\mathbb{S}^n}\int_{0}^{\infty}
2^{nt} \mathbb{I}_{ (\frac{f}{2+A})^n \leq 2^{−nt}}dtdV \leq
\frac{n16^n}{nA_n \log_e2} \int_{\mathbb{S}^n}\int_{0}^{(\frac{2+A}{f})^n }dsdV
\leq 
\frac{(16(2+A))^n}{A_n \log_e2}\int_{\mathbb{S}^n}(\frac{1}{f})^ndV $$
As a result, we have shown that the complexity is bounded by $\frac{(16(2+A))^n}{A_n \log_e2}\int_{\mathbb{S}^n}(\frac{1}{f})^ndV $

\begin{remark}
In fact the above algorithm can be generalised to a certain compact Riemannian manifold. We can generalise the initial notations in this section to any compact Riemannian manifold M and change the $\mathbb{S}^n$  appeared in the algorithm and lemmas to $M$, which gives the algorithm and the correctness proof.

The complexity analysis can be generalised to a certain compact Riemannian manifold, through the Bishop-G$\ddot{u}$nther inequality in\cite{8} which give the volume estimates of geodesic balls by curvature. The method is to use the Bishop-G$\ddot{u}$nther inequalities to give a new estimate of $\#divide(B)$ and follow the other part of the complexity analysis in the sphere. Notice we need to add some conditions on the compact Reimannian manifold to satisfy the requirement of this inequaltiy.

\end{remark}

\end{proof}
\end{subsection}

\end{section}

\begin{section}{Application to the homology of real semialgebraic set by homotopy equivalence
}

Notice we have designed an adaptive grid algorithm on the sphere. We will prove that $\frac{1}{ \kappa(F,G,x)}$ is a $D-Lipschitz$ function. As a result, ﬁnding the maximum of $\kappa(F,G,x)$ is a case on how to ﬁnd a minimum of a Lipschitz continuous function. Actually, we can assume that $\kappa(F,G,x)$ is positive, which will hold if there is no singular point. Then we will give a reformulation of the Niyogi-Smale-Weinberger theorem in a unit sphere. Finally we would like to apply our grid algorithm to generate a ﬁnite open balls whose union is homotopical to $S(F,G)$. As a result, the calculation of the homology group of semialgebraic set is equivalent to the calculation on the abstract simple set generated by these open balls since the homopoty equivalence implies the isomorphism of homology groups.

\begin{subsection}{Basic notions}
To a degree tuple $\mathbf{d} = (d_1,...,d_q)$ we associate it a linear space $\mathcal{H}_d[q]$ of polynomial systems $F = (f_1,...f_q)$ where $f_i  \in R[X_o,X_1,...,X_n]$ is homogeneous of degree $d_i$. There is a Euclidean inner product called Weyl inner product, deﬁned as follows. We have $h′ = \sum_{|\mathbf{a}|=d}h′_{\mathbf{a}}X{\mathbf{a}}$ in $R[X_0,...X_n]$, $\mathbf{a} = (a_0,...,a_n) \in  N^{n+1}$ where $|\mathbf{a}| := a0 + ... + an$. for homogeneous polynomials $h = \sum_{|\mathbf{a}|=d}h_{\mathbf{a}}X{\mathbf{a}}$, we have the deﬁnition.

\begin{definition}

$< h,h′ >= \sum \binom{\mathbf{d}}{\mathbf{a}}h_ah′_a
$
where $\binom{\mathbf{d}}{\mathbf{a}}$ denotes the multinomial coeﬃcient.
\end{definition}

For any $q$-tuples of homogeneous polynomials $F,F′ \in Hd[q]$ where $F = (f_1,...f_q)$ and $F′ = (f′_1,...f′_q)$, we have
$$< F,F′ >:= \sum^q _{j=1} < f_j,f′_ j >. $$
In other words, the Weyl inner product is a dot product with respect to a specially weighted monominal basis. In particular, it is invariant under orthogonal transformations of the homogeneous variables $(X_0,...,X_n)$. That is, for any orthogonal transformation $u : R^n \rightarrow R^n$ and any $f \in \mathcal{H}=_d[q]$, we have $||F|| = ||F \circ u||$. In all the later texts , all normes in spaces $\mathcal{H}_d[q]$ refer to the norm induced by the Weyl inner product.

For a point $x \in R^{n+1}$ and a system $F \in \mathcal{H}=_d[q]$ at $x \in  \mathbb{S}^n$ has been well studied. We deﬁne it as ∞ when the derivative $DF(x)$ of $F$ at $x$ is not surjective, otherwise as $$\mu_{norm}(F,x) := ||F||||DF(x)^+\Delta||$$, where the norm $||DF(x)^+\Delta||$ is the spectral norm. We also deﬁne the following variant of $\mu_{norm}$, more speciﬁc to homogeneous systems,
$$\mu_{proj}(F,x) := \mu_{norm}(F|\mathcal{T}_x)$$ where $T_x = {x}^{\perp}$
and $\mathcal{T}_x := x + T_x$. 
(The number $\mu_{norm}(F|\mathcal{T}_x)$ is well-deﬁned after identifying $\mathcal{T}_x$ with $R^n$.)

The numbers $\mu_{norm(F,x)}$ and $\mu_{proj(F,x)}$ measure the sensitivity of the zero $x$ of $F$ when $F$ is slightly perturbed. They are consequently useful at a zero, or near a zero, of the system $F$. To deal with points in Sn far away from the zeros of $F$, in particular to understand how much $F$ needs to be perturbed to make such a point a zero, a more global notion of conditioning is needed. 

\begin{definition}\label{Deﬁnition 3.2}
The real homogeneous condition number of $F \in \mathcal{H}_d[q]$ at $x \in \mathbb{S}^n$ is $\kappa(F,x) := ( \frac{1}{\mu_{proj}(F,x)^2}  + \frac{||F(x)||^2}{||F||^2}  )^{-\frac{1}{2}}$ , where we use the conventions $\infty^{-1} := 0, 0^{-1} := \infty$, and $\kappa(0,x) := \infty$. We further deﬁne $\kappa(F) := \max_{x\in\mathbb{S}^n }\kappa(F,x)$.

\end{definition} 
If $q > n$ (that is, if the system $F$ is overdetermined) then $DF(x)|\mathcal{T}_x$ cannot be surjective and $\kappa(F,x) = \frac{||F||}{||F(x)||} $ for all $x \in \mathbb{S}^n$ Thus, $\kappa(F) < \infty$ if and only if $F$ has no zeros in $\mathbb{S}^n$. The special case $F(x) = 0$ is worth highlighting. 

\begin{lemma}\label{Lemma 3.3.}
(\cite{1} Lemma 4.3.) For any $F \in \mathcal{H}_d[q]$ and $x \in\mathbb{S}^n$, if $F(x) = 0$, then
$\kappa(F,x) = \mu_{proj}(F,x) = \mu_{norm}(F,x).$
\end{lemma} 
We will introduce the following two propositions frequently since it makes the statement that take a positive $A-Lipchitz (A ≥ 1)$ continuous function $f : M \rightarrow (0,1]$ on the algorithms in the $§2$ reasonable. 

\begin{theorem}\label{Theorem 3.4. }
 (\cite{1} Corollary 4.5.) For any $F \in \mathcal{H}_d[q]$ and any $x \in \mathbb{S}^n, \kappa(F,x) \geq 1$. 

\end{theorem}
\begin{theorem}\label{Theorem 3.5. }
 (\cite{1} Proposition 4.7.) For $F \in \mathcal{H}_d[q]$ , the map $\mathbb{S}^n → \kappa(F,x)^{−1}$ is $D−Lipschitz$ continuous with respect to the Euclidean metric and sphere metric on $\mathbb{S}^n$.
\end{theorem}
Notice we are doing research on the semialgebraic set. As a result, we have to generalise the above deﬁnitions to the semialgebraic set.

We consider (closed) homogeneous semialgebraic systems, i.e., systems of the form
$f_1(x) = 0,...,f_q(x) = 0$ and $g1(x) \geq 0,...,g_s(x) \geq 0, ~~(\ast\ast\ast)$
where the $f_i$ and the $g_j$ are homogeneous polynomials in $R[X_0,X_1...,X_n]$. The system is a element $(F,G) \in \mathcal{H}_d[q + s]$. The set of solutions $x \in \mathbb{S}^n$ of system $(\ast\ast\ast)$, which we will denote by $S(F,G)$, is a spherical basic semialgebraic set. Needless to say, we do allow for the possibility of having $q = 0$ or $s = 0$. These correspond with systems having only inequalities (resp. only equalities.)

To a homogeneous semialgebraic system $(F,G)$ we associate a condition number $κ^{\ast}(F,G)$ and $κ(F,G,x)$ as follows. For a subtuple $L = (g_{j_1},...,g_{j_l})$ of $G$, let $F_L$ denote the system obtained from F by appending the polynimials from L, that is,
$F^L := (f_1,...,f_q,g_{j_1},...,g{j_l}) \in \mathcal{H}_d[q + l]$
(where now d denotes the appropriate degree pattern in $\mathbb{N}^{q+l}$ ). Abusing this notation, we will frequently use set notations $L \subset G$ or $g \in G $to denote subtuples or coeﬃcients of $G$.

\begin{definition}\label{Deﬁnition 3.6.}
Let $q \leq n+1$, $(F,G) \in \mathcal{H}_d[q + s]$. The condition number of the homogeneous semialgebraic system $(F,G)$ is deﬁned as
$$\kappa^{\ast}(F,G) := \max_{L\subset G, q +|L|\leq n+1}\kappa(F^L).$$
$$\kappa(F,G,x) := \max_{L\subset G, q +|L|\leq n+1}\kappa(F^L,x).$$
\end{definition} 

In addition, we need to generalise the D-Lipschitz condition from $\kappa(F)$ to the $\kappa(F,G,x)$. 
\begin{theorem}\label{Theorem 3.7.}
 
\end{theorem}
 For $(F,G) \in \mathcal{H}_d[q + s]$, the map $\mathbb{S}^n \rightarrow \kappa(F,G,x)^{−1}$ is $D −Lipschitz$ continuous with respect to the Euclidean metric and sphere metric on $\mathbb{S}^n$.
\begin{proof}
Take two points x,y ∈ Sn and suppose $\kappa(F,G,x) \geq \kappa(F,G,y) ≥ 0$. From the deﬁnition of $\kappa(F,G,x)$, there is a tuple $L \subset G $such that$\frac{1}{\kappa(F,G,x)} = 1\frac{1}{\kappa(F^L,x)}.$ We have $$| \frac{1}{\kappa(F,G,y)} - \frac{1}{\kappa(F,G,x)}| = \frac{1}{\kappa(F,G,y)} - \frac{1}{\kappa(F,G,x)} = \frac{1}{\kappa(F,G,y)} - \frac{1}{\kappa(F^L,x)}  $$ $$\leq  \frac{1}{\kappa(F^L,y)} - \frac{1}{\kappa(F^L,x)} \leq  D||x-y||\leq Dd_{\mathbb{S}}(x,y)$$
\end{proof}

For a nonempty subset $W \subset E$, let $d_W(x) := inf_{p\in W} ||x-p||$ the distance from $x$ to $W$.

\begin{definition}\label{Deﬁnition 3.8}
 The medial axis of $W$ is deﬁned as the closure of the set
$$\Delta_W := \{x \in E|\exists p,q \in W,p \neq q ~~and~~ ||x-p|| = ||x-q|| = d_W(x)\}$$. The reach or local feature size of W at a point $p \in W$ is deﬁned as$ \tau(W,p) := d_{\Delta_W} (p)$. The (global) reach of $W$ is deﬁned as $\tau(W) :=
inf_{p\in W}\tau(W,p)$. We also set $\tau(\emptyset) := +\infty.$
\end{definition}

In addition, $\tau(W)$ is given by $\inf_{x\in\Delta_W} d_W(x)$. That is a start point of our generalisation. Notice that we can also characterize $\tau(W)$ as the maximum of all $\epsilon$ such that for every $x \in E$ with $d_W(x) < \epsilon$, there exists a unique point $p \in W$ with $||-−p|| = d_W(x)$. We will denote this unique point by $\pi_W(x)$.
The \cite{1} deﬁnes $T(W) : \{x \in E|d_W(x) < \tau(W)\}$ and shows that 

\begin{theorem}\label{Theorem 3.9} 
(\cite{1} \textbf{Proposition 2.2.}) if $\tau(W) > 0$, then $\pi_W : T(W) \rightarrow W$ is continuous and the map
$$T(W)\times [0,1] \rightarrow T(W),~~(x,t) \rightarrow t\pi_W(x) + (1-t)x$$
is a deformation retract of $T(W)$ onto $W$.

\end{theorem}

There is a more general version of this theorem.

\begin{definition}\label{Deﬁnition 3.10}
 Given a compact set $W \subset E$. Now we deﬁne $T′(W) = E \setminus \Delta W$.

\end{definition}
Notice that $T′(W)$ is open in $E$ since $\Delta_W$ is closed.
Now we deﬁne a map from $T′(W)$ to $W$ as the following way. 

\begin{definition}\label{Deﬁnition 3.11}
 Since $\Delta_W$ is not in $T′(W)$, $\forall x \in T′(W)$, there is a unique point $p \in W$ such that $d(x,p) = d(x,W)$.Deﬁne $\pi_X : T′(X,U) \rightarrow W$ as $\pi_W(x) = p.$
\end{definition}
 
 \begin{theorem}\label{Theorem 3.12}
 If $\tau(W) > 0$, then $\pi_W : T′(W) \rightarrow W$ is continuous and the map $$T′(W)\times [0,1] \rightarrow T′(W),~~(x,t) \rightarrow t\pi_W(x) + (1-t)x$$
is a deformation retract of $T′(W)$ onto $W$. 
 \end{theorem}
  
The proof of this general version is quite similar to the \cite{1} \textbf{Proposition 2.2 .}
\begin{proof}
Concerning the continuity of $\pi_w$, let $(x_k)~~k\geq0$ be a sequence in $T′(W)$ converging to some $x \in T′(W)$. We have $$||\pi_W(x_k)-x||\leq||\pi_W(x_k)-x_k||+||x_k-x|| = d_W(x_k)+||x_k−x||\leq d_W(x)+2||x_k-x||$$,
where we used the Lipschitz continuity of $d_W$ for then last inequality. Hence the sequence $\pi_W(x_k)$ is bounded. Let $y \in W$ be a limit point of $\pi_w(x_k)$. The above inequality implies that $||y -x|| \leq d_W(x)$, hence $y = \pi_W(x)$. Thus $\pi_W(x)$ is the only limit point of the sequence $\pi_W(x_k)$ and therefore, $\lim_{k\rightarrow +\infty}\pi_W(x_k) = \pi_W(x)$.

The other statement is easy to be seen from the deﬁnition of the deformation retract.
\end{proof}

In addition, we introduce another kind of reach, which will be helpful in the following proofs. Let $W \subset E$ be a closed subset and $p \in W$. Moreover, consider $u \in E$ with $||u|| = 1$. It is easy to see that $\{t \geq 0|d_W(p + tu) = t\}$ is an interval containing $0$. We are interested in those directions $u$, where this interval has positive length and deﬁne the reach $\tau(W,p,u)$ of $W$ at $p$ along the direction $u$ as the length of this interval, that is,
$\tau(W,p,u) := \sup\{t \geq 0|d_W(p + tu) = t\}$. We note that $π_W(p + tu) = p$ for any $0 \leq t < \tau(W,p,u)$. And we have the following lemma.
\begin{lemma}\label{Lemma 3.13}
(\cite{1} \textbf{Lemma 2.5.}) Let $W \subset E$ be a closed subset, $p \in W$, and $u \in E$ be a unit vector such that $\tau(W,p,u)$ is positive. Then we have $\tau(W,p) \leq \tau(W,p,u)$.
\end{lemma}

Now we will give a more detail analysis about the neighbourhoods of spherical basic semialgebraic sets. There are two kinds of neighbourhoods in our discussions. For a subset $A \subset \mathbb{S}^n$ we denote by $\mathcal{U}_{\mathbb{S}}(A,r) := \{x \in \mathbb{S}^n|d_{\mathbb{S}}(x,A) < r\}$ the open r−neighborhood of $A$ with respect to the geodesic distance $d_{\mathbb{S}}$ on the sphere $\mathbb{S}^n$.
Also, for a homogeneous system $(F,G) \in \mathbb{H}_d[q+s]$ and $r > 0$, we deﬁne the $r−relaxation$ of $S(F,G)$:
$$Approx(F,G,r) := {x ∈\mathbb{S}^n|\forall f \in F |f(x)| < ||f||r~~ and ~~g \in G,~~ g(x) > −||g||r}.$$

\end{subsection}

\begin{subsection}{ An extension of the Niyogi-Smale-Weinberger theorem}

In the \cite{1}, the authors observed an extension of the Niyogi-Smale-Weinberger theorem on any compact subset $X$, $S$ provided $S$ has positive reach $\tau(S)$. By deﬁning Hausdorﬀ distance between two nonempty closed subsets $A,B \subset E$ as
$$d_H(A,B) := \max(\sup_{a\in A} (d_B(a)),\sup_{ b\in B} d_A(b)),$$
the theorem can be stated as
\begin{theorem}\label{Theorem 3.14}
(\cite{1} \textbf{Theorem 2.8.}) Let $S$ and $X$ be nonempty compact subsets of $E$. The set $S$ is a deformation retract of $ U(X,\epsilon)$ for any $\epsilon$ such that $3d_H(S,X) < \epsilon < 1 2\tau(S).$
\end{theorem}

However, for the purpose of a more adaptive algorithm, we would like a local version of this theorem. If we get the local version, we can design an recursion algorithm by subdividing a set $X$ to small open covers. Fortunately, we can remove the Hausdorﬀ distance and give a local version of Niyogi-Smale-Weinberger theorem.

\begin{theorem}\label{Theorem 3.15}
Let $S$ and $X$ be two nonempty compact subsets of $E$. Given a function $\epsilon : X \rightarrow \mathbb{R}$. Suppose $$\forall x \in X,~~d(x,\Delta S) > 5\epsilon(x)$$, $$\forall x \in X,~~ d(x,S) <\frac{1}{3}\epsilon(x),$$
$$\forall p \in S, \exists x \in X ~~such~~that~~ ||x-p|| <
\frac{1}{4}\epsilon(x)$$
and $$\forall x,y \in X, ~~|\epsilon(x)-\epsilon(y)| <
\frac{1}{ 17}||x-y||,$$ we get that the set $S$ is a deformation retract of $\mathcal{\mathcal{U}}(\mathcal{X},\epsilon)$.
\end{theorem}
\begin{proof}
At ﬁrst, we prove the main theorem. For the convenience of proof, we let $A_1 = 5,A_2 =\frac{1}{3},A_3 =\frac{1}{4},A_4 =\frac{1}{17}$ for each coeﬃcient in the above equalities. Take $u \in  \mathcal{U}$, $\exists x \in  \mathcal{X}$ such that $$d(x,p) < \epsilon(x),~~ d(p,\Delta_S) > d(x,\Delta S)−d(x,p) > (A_1)\epsilon(x) > \epsilon(x)$$. 
So $\mathcal{U}\cap \Delta_S = \emptyset$ which implies that $\pi_S$ is well deﬁned on $\mathcal{U}$.

As a result, we can see that the map
$$\mathcal{U}(\mathcal{X},\epsilon)\times[0,1] \rightarrow E,~~ (x,t) \mapsto t\pi_S(x) + (1−t)x$$
is well deﬁned. The map is also continuous by the above theorem. It remains to prove that its image is included in $\mathcal{U}(\mathcal{X},\epsilon)$. Let $v \in B(\mathcal{X},\epsilon(x))$ and $p := \pi_S(v)$. If $||p−x|| \leq \epsilon(x)$, then the line segment $[v,p]$ is entirely included in the ball of radius $\epsilon(x)$ around $x$, which is a part of $\mathcal{U}$, and we are done. So we assume that $||p−x||\geq \epsilon(x)$. let $u$ be the unique point in $[v,p]$ such that $||u-x|| = \epsilon(x)$. The line segment $[v,u)$ being included in the ball $B(x,\epsilon(x)) \subset \mathcal{U}$, it only remains to check that $[u,p]$ is also included in $\mathcal{U}$. Let $r :=\frac{1}{3\epsilon(x)}$. Also, let $l$ be the open half line starting from $p$ and passing through $v$ and $w$ be the unique point in $l$ such that $||w -p|| = 2\epsilon(x)$. 

At ﬁrst, we want to show that $d(p,\Delta_S) > 2\epsilon(x)$. Notice that
$$||x-p||≤||x-v||+||v-p|| = d(v,S)+||x-v||\leq d(v,x)+d(x,S)+\epsilon(x) < (2+A_2)\epsilon(x)$$
and
$$d(p,\Delta_S) \geq d(x,\Delta_S)−d(x,p) > (A_1 −A_2 −2)\epsilon(x) > 2\epsilon $$. Also, as $p = \pi_S(v) \in S_x$, we have $\tau(S,p,\frac{ w-p }{||w-p||}  ) > 0$. By \cite{1} \textbf{lemma 2.5}, we obtain $\tau(S,p) \leq \tau(S,p,\frac{ w-p }{||w-p||})$ and therefore $2\epsilon(x) < \tau(S,p, \frac{ w-p }{||w-p||} )$. This implies that $\pi_S(w) = p$ and $d_S(w) = ||w-p|| = 2\epsilon(x)$. Let $q = \pi_S(x)$. Since $||x-q|| = d(x,S) < A_2\epsilon(x)$, $$||w-x||\geq ||w-q||-||q-x|| > d(w,S)-A_2\epsilon(x) \geq (2 + A_2)\epsilon(x) \geq\frac{5}3 \epsilon(x).$$

By our assumption, we can ﬁnd $y \in X$ such that $||y -p|| \leq A_3\epsilon(x)$. Now we want to show $d(y,u) < \epsilon(y)$. $A_3 < 1$ implies that $||p-y|| < A_3\epsilon(y) < \epsilon(y)$. Notice $$||y-u||\leq||y-p||+||p-u|| < A_3\epsilon(y)+||w-p||-||w-u|| = A_3\epsilon(y)+2\epsilon(x)-||w-u||,$$ $$||w-u||2 \geq||w-x||^2 -||x-u||^2 > (\frac{5}{3}\epsilon(x))^2 −(\epsilon(x))^2 = \frac{16}{9}\epsilon(x)^2 \implies ||w-u|| > \frac{4}{3}\epsilon(x).$$ Now we get
$$||y-u|| < A_3\epsilon(y) + 2\epsilon(x)−
\frac{4}{3}\epsilon(x) = A_3ϵ(y) +
\frac{2}{3}\epsilon(x).$$
By our assumption, we can get
$$||x-y||\leq||y-u||+ \epsilon(x) \leq A_3\epsilon(y) +
\frac{2}{5}\epsilon(x)$$ $$\leq A_3\epsilon(y) + \frac{5}{3}
(\epsilon(x)−\epsilon(y)) +
\frac{5}{3}\epsilon(y) \leq A_3\epsilon(y) +
\frac{5}{3}A_4||x−y||+
\frac{5}{3}\epsilon(y).$$
As a result, we get
$$||x-y||\leq\frac{1- \frac{5}{3}A_4}{ A_3 +\frac{5}{3}}\epsilon(y).$$
Finally, we get
$$||y−u|| < A_3\epsilon(y) +
\frac{2}{3}
(\epsilon(x)−\epsilon(y)) +
\frac{2}{3}
\epsilon(y) \leq (
\frac{2}{3}
+ A_3)\epsilon(y) +
\frac{2}{3}
A_4||x−y||
\leq (
\frac{2}{3}
+ A_3 +
\frac{2}{3}
A_4
 \frac{A_3 + \frac{5}{3}}{1-\frac{5}{3}A_4}
)\epsilon(y)  \leq \epsilon(y).
 $$
\end{proof}

\end{subsection}

\begin{subsection}{Refoumulation of the Niyogi-Smale-Weinberger theorem in sphere}

Now we will establish two theorems to get a reformulation of the Niyogi-Smale-Weinberger theorem in a unit sphere, which will be a basis for our main algorithm.

\begin{subsubsection}{Local versions of the \cite{1} \textbf{theorem 4.12} and \cite{1} \textbf{Proposition 4.17}
}

At ﬁrst, let us recall the \cite{1} Theorem 4.12 and \cite{1} Proposition 4.17 which are important in the proof of \cite{1} Proposition 5.1.

\begin{theorem}\label{Theorem 3.16.}
(Theorem 4.12 in [1]) For any homogeneous semialgebraic system $(F,G)$ deﬁning a semialgebraic set $S := S(F,G) \subset \mathbb{S}^n$, if $\kappa^{\ast}(F,G) < \infty$, then
$$D^{\frac{3}{2}} \tau(S)\kappa^{\ast}(F,G) \geq \frac{1}{7}$$
\end{theorem}

\begin{theorem}\label{Theorem 3.17.}
 (Proposition 4.19 in \cite{1}) Let $q \leq n + 1$. For any positive number $r < (13D^{\frac{3}{2}} κ^2_{\ast})^{−1}$ we have $$Approx(F,G,r) \subset  \mathcal{U}_S(S(F,G),3κ_{\ast}r).$$
\end{theorem}
  
Notice the two main theorems are based on the global condition numbers of a set. Actually, we can get local versions of the two main theorem by changing $κ_{\ast}$ with $\kappa(F,G,x)$. Notice we have gotten a local version of the Niyogi-Smale-Weinberger theorem. We will show combining the above theorems leads to a more adaptive local version of NiyogiSmale-Weinberger theorem in a unit sphere.

Now I will give a local version of \textbf{Theorem 3.16}. We begin by two lemmas, which are local versions of \textbf{Theorem 2.4} and \textbf{Corollary 2.6} in \cite{1}.
\begin{theorem}\label{Theorem 3.18}
 For closed subsets $V,W$ of $E$ we have $\tau(W\cap V,p) \geq \min(\tau(W,p),\tau(W\cap \partial V,p))$.
\end{theorem}
\begin{proof}
Since $W \cap V \subset W$ and $W \cap V \subset W \cap \partial V$, we have $\Delta W\cap V \subset \Delta W$ and $\Delta W\cap V \subset \Delta W\cap \partial V $. Now we ﬁnish the proof by the deﬁnition of $\tau(W,p)$.
\end{proof}
\begin{corollary}\label{Corollary 3.19.}
For closed subsets $V_1,...,V_s$ and $W$ of $E$ we have
$\tau(W \cap V_1 \cap...\cap V_s,p)\geq  \min_{I\subset(1,...,s)}\tau(W \cap \bigcap_{\in I}\partial V_i,p)$
\end{corollary}

\begin{proof}
The case $s = 1$ is covered by the above theorem. In general, we argue by induction on $s$. $$\tau(W \cap V_1 \cap...\cap V_{s+1},p) ≥ \min(\tau(W \cap V_1 \cap...\cap V_s,p),τ(W \cap V_1 \cap...\cap \partial V_s,p))
$$ $$\geq \min( \min_{I\subset(1,...,s)}
\tau(W \cap\bigcap_{i\in I}
\partial V_i,p), \min_{I\subset (1,...,s)}
\tau(W \cap \partial V_{s+1} \cap \bigcap_{i\in I}
\partial V_i,p))
= \min_{I\subset(1,...,s+1)}
\tau(W \cap \bigcap_{ i\in I}\partial
V_i,p)$$
where we have applied the above theorem and twice the induction hypothesis.
\end{proof} 

Now we can get a local version of \textbf{theorem 3.16}. Although in the \textbf{theorem 3.16} we need to take $x \in S(F,G)$, we can get a general inequality when $x \notin S(F,G)$.

\begin{theorem}\label{Theorem 3.20.}
For any semialgebraic system $(F,G)$ deﬁning a semialgebraic set$ S : S(F,G) \subset\mathbb{S}^n$, if $κ^{\ast}(F,G) < \infty$, then $\forall x \in S(F,G)$ we have
$$D^\frac{3}{2}\tau(S,x)\kappa(F,G,x) \geq\frac{1}{7} ,~~ \forall x \in S(F,G).$$
What is more, $\forall x \in \mathbb{S}^n$ we have $$dS(x,∆S) \geq
\frac{1}{
7D^\frac{3}{2}\kappa(F,G,x)} -
\frac{8}{ 7}d_{\mathbb{S}}(x,S(F,G))$$

\end{theorem}
\begin{proof}
From the proof of \cite{1} \textbf{Theorem 4.12}, we get for any homogeneous algebraic system $(F)$ deﬁning a semialgebraic set $S : S(F) \subset \mathbb{S}^n$,
if $\kappa^{\ast}(F) < \infty$, then
$$D^\frac{3}{2}\tau(S,x)\kappa(F,x) \geq
\frac{1}{7}.$$

We turn to the general case $S := S(F,G) \subset \mathbb{S}^n$ and we assume $\kappa^{\ast}(F,G) < \infty.$ For $g \in G$ we deﬁne $P_g := \{x \in \mathbb{S}^n|g(x) \geq 0\}$ and $W := S(F,\emptyset)$ so that $S = W \cap(\cap_{g\in G}P_g)$. We claim that for any $L \subset G$, $$W \cap(\cap_{g\in G}P_g) = S(F^L,\emptyset)$$. The left to right inclusion is clear since $\partial P_g$ is contained in the zero set of $g$. Conversely, let $x \in S(F^L,\emptyset)$ (in particular, $q+|L|\leq n$, by \cite{1}). The derivative $DF^L(x)$ is surjective, because $κ(F^L,x) < \infty$. In particular, for any $g \in L$, $Dg(x) \neq 0$ and since $g(x) = 0$ it follows that the sign of $g$ changes around $x$. Thus $x \in \partial P_g$ and the above equation follows.

The above theorem implies that
$$\tau(S,x) \geq \min_{L\subset G}
\tau(W \cap(\cap_{g\in G}P_g) ) = \min_{L⊂G}
\tau(S(F^L,\emptyset,x)).$$
It suﬃces to take the minimum over the $L \subset G$ such that $q + |L| \leq n + 1$ because $S(F^L,\emptyset) = \emptyset$ for larger $L$. We obtain from the case $G = \emptyset$ above,
$$7D^{
\frac{3}{2}} \tau(S,x) \geq \min_L
7D^{\frac{3}{2}} \tau(S(F^L,\emptyset),x) \geq \frac{1}{ \max_L κ(F^L,x)}
=
\frac{1}{ \kappa(F,G,x)}
.$$
Now take $y \in S$ such that $d_{R^{n+1}}(x,S) = d_{R^{n+1}}(x,y)$. Then $$d_S(x,∆S) \geq \tau(S,x) \geq d_{R^{n+1}}(y,\Delta S)-d_{R^{n+1}}(x,y) \geq
\frac{1}{ 7\kappa(S(F,G),y)} -d_{R^{n+1}}(x,y)
$$ $$\geq
\frac{1}{ 7\kappa(S(F,G),x)} -
\frac{D}{ 7D
^{\frac{3}{2}}}
d_{R^{n+1}}(x,y)-d_{R^{n+1}}(x,y) \geq
\frac{1}{ 7D
^{\frac{3}{2}}} \kappa(F,G,x) -
\frac{8}{7}
d_S(x,S)$$
from the Lipchitz continuity of $\frac{1}{\kappa(F,G,x)}$.
\end{proof}

Let us recall the \cite{1} \textbf{Theorem 4.19} which related two neighbourhoods of a set in an Euclidian space.
\begin{theorem}\label{Theorem 3.21.}
(\cite{1} \textbf{Theorem 4.19}.) Let $q \leq n + 1$. For any positive number $r < (13D^\frac{3}{2}\kappa^2_{\ast})^{−1}$ we have
$Approx(F,G,r) \subset\mathcal{U}_{\mathbb{S}}(S(F,G),3\kappa_{\ast}r)$.

\end{theorem}
 
Actually, there is local version of the above theorem, which will be introduced after the following lemma, which is a local version of \cite{1} \textbf{Lemma 4.18.}

\begin{lemma}\label{Lemma 3.22.}
Let H ⊂ L ⊂ G be such that |H| = n−q + 1. Suppose that κ(FH) < ∞ and 0 < r < 1 κ(FH,p), p ∈Sn. Then p / ∈ Approx(FL,G\L,r).
\end{lemma}
 \begin{proof}
 Since $\kappa(F^H) < \infty$ we have $S(F^H,\emptyset) = \emptyset$, by \textbf{Lemma 4.11}in \cite{1}. Assume $p \in Approx(F^L,G\setminus L,r)$. Then as $H \subset L$ we have that $|h(p)|\leq e||h||$. for all $h \in F^H$ and it follow
 s that $$\frac{1}{\kappa(F^H,p)}  = \frac{||F^H(p)||}{ ||F^H||} \leq r$$. This contradicts with the hypothesis on $r$ and hence $p \notin Approx(F^L,G\setminus L,r)$.
 \end{proof}

Now we can give the proof of the local version of \ref{Theorem 3.21.}.

\begin{theorem}\label{Theorem 3.23.}
Deﬁne $B_s = ∑^{s−1}_{ i=o}(4D)^i = \frac{(4D)^s−1}{ 4D−1}$. Assume that $\kappa(F,G) < \infty$. Let $q \leq n + 1$. For any positive number $r < 13D 3 2 κ(F,G,x)$ we have $$x ∈ Approx(F,G,r) =⇒ dS(x,S(F,G)) < 3Bsκ(F,Gx)r.$$
\end{theorem} 
\begin{proof}
We will abbreviate $S := S(F,G)$ and $\kappa(x) := \kappa(F,G,x)$. The proof is by induction on $s$ which is the number of polynomials in $G$.

We use $l = n−q + 1$ to denote the diﬀerence between the number of variables and the number of equations. If $l = n−q + 1 = 0$, then $\kappa(F,p) < \infty$ and because of our hypothesis, $r <\frac{1}{\kappa(F,p)}$. We deduce from the \textbf{lemma 3.22} with $L = H = \emptyset$. Now we assume $l > 0$, i.e. $q \leq  n$, and consider a point $x \in Approx(F,G,r)$. It is enough to show that $d_S(x,S) < 3\kappa(x)r$.

To do so, we focus on the set 
$L := \{g \in G | ~|g(x)| < r||g||.\}$ By construction, we have $x \in Approx(F^L,G\setminus L,r)$, and moreover $g(x) \geq r||g|| > 0$ for all $g \in G\setminus L$. We further note that $|L|\leq n−q$, otherwise there would exist $H \subset L$ with $|H| = n−q+1$ and we would use again \textbf{lemma 3.22} to deduce that $Approx(F,G,r) = \emptyset$, in contradiction with the fact that $x \in Approx(F,G,r)$. We next divide by cases. 

\textbf{Case 1} : $L \neq \emptyset$. As $|F^L|\leq n + 1$ we may apply the induction hypothesis to the larger set $F^L$ of equations and the smaller set $G\setminus L$ of inequalities. Note that $\kappa(F^L,G\setminus L,x) \leq \kappa(F,G,x)$ so the hypothesis on $r$ is still true for $S(F^L,G\setminus L)$. The inclusion hypothesis yields
$$p \in Approx(F^L,G\setminus L,r) \implies d_S(x,S(F^L,G\setminus L)) < 3B_{s−1}\kappa(FL,G\L,x)r \implies d_S(x,S) < 3B_s\kappa(x)r.$$Hence we obtain the theorem in this case.

\textbf{Case 2} : $L = \emptyset$. We put $u := \frac{|F(x)|}{||F||}$. Then $u \leq r$ since $x \in Approx(F,G,r)$. Moreover, $\kappa(F,G,x)u \leq \kappa(F,G,x)r < 1 3$ by the assumption. By deﬁnition, $\kappa(F,G,x) \geq \frac{1}{ \mu_{proj}(F,x)^{−2}+u^2} \geq \frac{1}{ 2}\min\{\mu_{proj}(F,x)^2,u^{−2}\}.$ The minimal euqals $\mu_{proj}(F,x)^2$ since $\kappa(F,G,x)u \leq \frac{1}{13}$. So we get $\sqrt{2}\kappa(F,G,x) \geq \mu_{proj}(F,x) = \mu_{norm}(\widetilde F,x)$ where $\widetilde F := F|T_x$ denotes the restriction
of $F$ to the affine space $T_x$. It follows that $$alpha (\widetilde F,x) \leq \frac{1}{2}D^{\frac{3}{2}} \mu_{norm}(\widetilde F,x)2u$$ $$ \leq D^{\frac{3}{2}}  \kappa(F,G,x)2r \beta(\widetilde F,x) \leq \mu_{norm}(\widetilde F,x) \leq\sqrt{2}\kappa(F,G,x)r$$.

From the assumption on $r$, we get $\alpha(\widetilde F,x) \leq \frac{1}{13}$ which makes possible the application of \textbf{Theorem 3.1}. We also note that $\beta(\widetilde F,x) < \frac{1}{13}$. As in \textbf{section 3.2} of \cite{1}, we deﬁne $x_t$ in the affine space $T_x$ by the system of diﬀerential equations
 $$\dot{x_t} = −D\widetilde F(x_t)^+\widetilde F(x_t), x_0 = x$$. Note that $x_t \neq 0$ for all $t \geq 0$ as $||z||\geq 1$ for all $z\in T_x$. We deﬁne $y_t := \frac{x_t}{||x_t||}\in \mathbb{S}^n$. By \textbf{theorem 3.1}, there is a limit
 point $x_{\infty} \in T_x$, which is a zero of $\widetilde F$ which satisﬁes
 $||x_{\infty}−x|| < 2\beta(\widetilde F,x)$. In particular, $y_{\infty}$ is a zero of $F$ and $d_S(y_{\infty},x) \leq ||x_{\infty}−x||\leq 2\beta(\widetilde F,x) \leq 2\sqrt{2}\kappa(F,G,x) < 3\kappa(F,G,x)$ If $g(y_{\infty}) \geq 0$ for all $g \in G$, then $y\infty \in S$ and $d_S(x,S) ≤ d_S(x,y_{\infty})$ hence we are done.
 
Notice we have proved that for any positive number $r < (13D^{\frac{3}{2}}  \kappa(F,x)^2)^{-1}$, if $x \in  Approx(F,r)$, then $d_S(S(F),x) < 3\kappa(F,x)r$. So we prove the initial case of the induction when $s = 0$.

So suppose that $g(y_{\infty}) < 0$ for some $g \in G$ and let $s > 0$ be the smallest real number such that $g(y_s) = 0$ for some $g \in G$. By construction, the set $H := \{g \in G|g(y_s) = 0\}$ is nonempty and element of $G\setminus H$ is positive at $y_s$. Also, for every $f \in F$,
$$|f(y_s)| =
\frac{f(x_s)}{ ||x_s||^{deg(f)}} \leq |f(x_s)| = |f(x)|e^{-s} \leq ||f||re^{-s}$$ where the second equality is due to \textbf{Theorem3.1.(i)} in \cite{1}. Therefore, $y_s \in Approx(F^H,G\setminus H,re^{-s})$
.
Using again the \ref{Lemma 3.22.} we deduce that $|H| < n-q+1 = l$. We can therefore apply the induction hypothesis to the larger set $F^H$ of equations and the smaller set $G\setminus H$ of inequalities. Thus we obtain
$y_s \in Approx(F^H,G\setminus H,re^{-s})$ $$\implies d_S(y_s,S(F^H,G\setminus H)) < 3B_{s-1}\kappa(F^H,G\setminus H,y_s)re^{-s} \implies d_S(y_s,S(F,G)) < 3B_{s-1}\kappa(y_s)re^{-s},$$
the latter because $S(F^H,G\setminus H) \subset S$ and $A_{s-1}\kappa(F^H,G\setminus H,y_s) \leq B_{s-1}\kappa(y_s)$. Also, by \textbf{theorem 3.1(ii)} in \cite{1}, $$d_S(y_s,x) \leq ||x_s -x||\leq \beta( \widetilde F,x)(1-e^{-s}) < 2\sqrt{2}\kappa(x)r(1-e^{-s}).$$
We ﬁnally deduce that
$$d_S(x,S) \leq d_S(x,y_s) + d_S(y_s,S) < 3(1-e^{-s})\kappa(x) + 3B_{s-1}e^{-s}\kappa(y_s)$$ $$ \leq 3\kappa(x) + 3B_{s-1}D\kappa(x)\kappa(y_s)||x-y_s|| \leq 3(1 + 4DB_{s-1})\kappa(x) = 3B_s\kappa(x).$$
\end{proof}

\end{subsubsection}

\begin{subsubsection}{Main theorem}
Notice the fourth condition in our extension of the Niyogi-Smale-Weinberger theorem (\ref{Theorem 3.15}) is a Lipschitz continuity condition about the radius function. Since $\frac{1}{ \kappa(F,G,x)}$ is $D$-Lipschitz continuous on sphere (\ref{Theorem 3.7.}) , it is nature to consider the $\epsilon(x) = \frac{c}{\kappa(S(F,x))}$ for certain constant $c$. In this case, we can simplify the requirements of our Niyogi-Smale-Weinberger theorem. 
\begin{corollary}\label{Corollary 3.24.}
 Let $S$ and $\mathcal{X}$ be two nonempty compact subsets of $\mathbb{S}^n$. Given a function $\epsilon : \mathcal{X} → R$. Suppose $\epsilon(x) = \frac{c}{ \kappa(F,G,x)}$. If
$$c ≤
\frac{1}{ 40D^{\frac{3}{ 2}}},$$
then
$$\max_{f\in F,g\in G}\{\frac{||f(x)||}{ ||f||}, \frac{−g(x)}{ ||g|| }\} < \frac{1}{ 13(4D)^sD^{\frac{3}{ 2}} \kappa^2(F,G,x)} $$
and
$$\forall p \in S, \exists x ∈ \mathcal{X}~~ such~~ that~~ d_S(x,p) <
\frac{1}{ 4}
\epsilon(x)$$
are enough to show that the set $S$ is a deformation retract of $\mathcal{U}(\mathcal{X},ϵ)$.

\end{corollary}

Now we give a proof of our corollary. 

\begin{proof}
 Let $B_s =\sum^{s−1}_{ i=o}(4D)^i = \frac{(4D)^s−1}{ 4D−1}$
 in this corollary. Notice that $B_s < (4D)^s$ since
$D \geq 1$.
Take $B_S(x,r(x))$ such that $$\max\{\frac{||f(x)||}{ ||f||}
, \frac{−g(x)}{ ||g||} \} < r(x) <\frac{ 1}{ 13(4D)^s D^{\frac{ 3}{ 2}} k^2(F,G,x)} \leq\frac{ 1}{ 13B_s D^{\frac{ 3}{ 2}} \kappa^2(F,G,x)}.$$
Then $x \in Approx(F,G,D^{\frac{1}{ 2}} r(x))$ and $13D^{
5}{ 2} \kappa^2(F,G,x)r(x) < 1$ by the assumption. Hence by the \textbf{theorem 3.23}, $$d_S(x,S) < 3_Bs\kappa(F,x)r(x)D^{\frac{ 1}{ 2}} < 3 13D^{\frac{ 5}{ 2}} \kappa(F,G,x)\leq \frac{1}{ 3\epsilon(x)}.$$
Now by the above theorem, we get
$$d_S(x,∆_S) \geq
\frac{1}{
7D^{\frac{
3}{ 2}} \kappa(F,G,x)} -
\frac{
8}{7}
d_S(x,S(F,G)) >
\frac{1}{
7D^{\frac{
3}{ 2}} \kappa(F,G,x)} -
\frac{
8}{21}
\epsilon(x) > 5\epsilon(x)$$
Notice that it is easy to ﬁnd that the second condition in our corollary implies the fourth assumption in our main theorem. Notice
$$c(
\frac{1}{ \kappa(S(F,G),x)} -
\frac{1}{ \kappa(S(F,G),y)}
) \leq cD||x-y|| <
\frac{1}{ 40D^{\frac{
1}{ 2}}}||x-y|| <
\frac{1}{17}||x-y||$$
since $D \geq 1$.

\end{proof}

\end{subsubsection}

\end{subsection}

\subsection{Algorithm for the computation of homology on a semialgebraic set.}
Now let us design a more adaptive algorithms than the main algorithm in the proposition 5.1 of \cite{1}. We will use those notations of \S 2.3.

\begin{theorem}\label{Theorem 3.25.}
For any semialgebraic set $S(F,G)$ where $\#F = q,\#G = s$, there is an algorithm to generate a finite set of balls $\mathcal{B}$ such that
$\bigcup_{B\in\mathcal{B}}B_{\mathbb{S}^{n}}\left ( c_{B},r_{c_{B}} \right )$ is homotopical to $S(F,G)$.  The complexity is $ \dfrac{\left( nsD^{-s}\right)
^{\mathcal{O}_{(n)}}}{A_{n}\log_{e}{2}} \int_{\mathbb{S}^{n}} \left ( \kappa \left (  F,G,x \right ) \right )^{n}dV$ where $A_{n}$ is the area of a $n-1$ dimension unit sphere. 
\end{theorem}

\begin{algorithm}  
        \caption{ Covering(F,G) }  
        \begin{algorithmic}[1] 
            \Require  homogeneous algebraic system (F,G)
            \Ensure    ﬁnite set $\mathcal{X} = {B_{\mathbb{S}^n}(x,\epsilon(x)}$ of balls such that $\cup B_{\mathbb{S}^n}(x,\epsilon(x)$ is homotopical to $S(F,G)$
            \State $\mathcal{X}  \gets \emptyset$
            \State $Cover( \frac{1}{ 360(4D)^sD^{ \frac{5}{2}} κ(F,G,x)^2}$) 
            \State $L_{final} \gets \{B(x,r_x)|B(x,r_x) \in L_{final}~~ and~~ x \in Approx(F,D, \frac{1}{2}r_x)\}$
            \For{$B_{\mathbb{S}^n}(c_B,r_{c_B}) \in L_{final}$ }
            \If{ $c_B \in Approx(F,G,D^{\frac{1}{2}}r(c_B))$ }
                        \State  $\epsilon(c_B) \frac{1}{40D^{\frac{3}{2}}\kappa(F,G,c_B)}$            
                        \State Add $B(c_B,\epsilon(c_B))$ to $\mathcal{X}$. 
                    \Else 
                      \State Pass
            \EndIf                       
           \EndFor
          \State \Return { $\mathcal{X}$}
        \end{algorithmic}  
    \end{algorithm}  
\begin{proof}
Since $\kappa\left( F,G,x\right) \geq 1$ and $r$ tends to $0$, we see that this algorithm will terminate. So the main problem is the correctness. Suppose $\mathcal{B}$ is the final output of our algorithm. Notice $\epsilon\left( x \right)$ satisfies the requirement of the radius function in the \ref{Corollary 3.24.}.
\\
\\
At first, let $B_{\mathbb{S}}\left(x,r\left(x\right)\right)\in \chi$. Then $x \in Approx\left(F,G,D^{\frac{1}{2}}r\left(x\right)\right)$ and $13\left(4D\right)^{s}D^{\frac{5}{2}}\kappa^{2}\left(F,G,x\right)r\left(x\right) < 360\left(4D\right)^{s}D^{\frac{5}{2}}\kappa^{2}\left(F,x\right)r\left(x\right) < 1$. By combining the above two inequalities, we obtain the condition $\left(i\right)$ of our \textbf{Corollary 3.24.}.
\\
\\
Finally, let us prove that it satisfies the condition $\left(ii\right)$ of \textbf{Corollary 3.24}. By theorem 4.17 in \cite{1}, for any $y \in S(F,G)$, there is a ball $B_{\mathbb{S}}\left(x,r\left(x\right)\right) \in \mathcal{B}$ such that $d_{\mathbb{S}}\left(x,y\right) < r\left(x\right)$ and  $x \in Approx\left(F,G,D^{\frac{1}{2}}r(x)\right)$. Because $360D^{\frac{5}{2}}\kappa^{2}\left(F,x\right)r\left(x\right) < 1$, $\kappa(F,G,x) \geqslant 1$ and $D \geq 1$, we get $r(x) < \frac{1}{360D^{\frac{5}{2}}\kappa(F, G, x)} < \frac{1}{180D^{\frac{3}{2}}\kappa(F, G, x)} = \frac{1}{4}\epsilon(x)$.
\\
\\
Now, by the \textbf{Theorem 2.6}, we get the complexity is bounded by $\dfrac{\left(16\left(2+A\right)\right)^{n}}{A_{n}\log_{e}2}\int_{\mathbb{S}^{n}}\left(\frac{1}{f}\right)^{n}dV$ if we omit the computation of $\kappa\left(F^{L},x\right)$. Notice the Lipchitz constant of the function $\frac{1}{360\left(4D\right)^{s}D^{\frac{5}{2}}\kappa^{2}\left(F,G,x\right)}$ is $\frac{1}{180(4D)^{s}D^{\frac{3}{2}}}$. As a result, $A = \frac{1}{180(4D)^{s}D^{\frac{3}{2}}}$.
\\
\\
In \cite{4} \textbf{2.5} the author can show that the complexity of computing $\kappa\left(F^{L},x\right)$ can be approximated within a factor 2 in $\mathcal{O}\left(N + n^{3}\right) = \left(nD\right)^{\mathcal{O}\left(n\right)}$ operations. In addition, in each calculation of $\kappa(F,G,x)$, we should calculate $\kappa(F^{L},x)$ within times $M = \sum_{i=0}^{n+1-q}\binom{s}{i}\leq (s+1)^{n+1-q}$. As a result, we get that the complexity of the algorithm in the last corollary is less than $\dfrac{\left(16nDs\left(2+A\right)\right)^{\mathcal{O}(n)}}{A_{n}\log_{e}2} \int_{\mathbb{S}^{n}} \left ( \kappa \left (  F,G,x \right ) \right )^{n}dV = \dfrac{\left(ns\left(4D\right)^{-s}D^{-\dfrac{1}{2}}\right)^{\mathcal{O}(n)}}{A_{n}\log_{e}2} \int_{\mathbb{S}^{n}} \left ( \kappa \left (  F,G,x \right ) \right )^{n}dV = \dfrac{\left(nsD^{-s}\right)^{\mathcal{O}(n)}}{A_{n}\log_{e}2} \int_{\mathbb{S}^{n}} \left ( \kappa \left (  F,G,x \right ) \right )^{n}dV$ from \textbf{Theorem 2.6}.
\end{proof}

Once in the possession of a pair $(\chi,\epsilon)$ such that $S$ is a deformation retract of $\mathcal{U}(\chi,\epsilon)$, the computation of the homology groups of $B_{\mathbb{B}}(\chi,\epsilon)$ is a known process. One can compute the nerve $\mathcal{N}$ of the covering ${B(x,\epsilon)\mid x \in \chi}$  (this is the simplicial complex whose elements are the subsets $\mathcal{N}$ of$\chi$ such that $\cap_{x\in N}B(x,\epsilon)$ is not empty) and from it, its homology groups $H_{k}(\mathcal{N})$. Since the intersections of any collection of balls is convex, the Nerve Theorem (Corollary 4G.3 of \cite{4}) ensures that
\[
H_{k}(\mathcal{N}) = H_{k}\left(\mathcal{U}\left(\chi,\epsilon\right)\right) = H_{k}(S).
\]
The last is because that $S$ is a deformation retract of $\mathcal{U}(\chi,\epsilon)$.
\\
The process is described in detail in \cite{5} \S 4 where the proof of the following result can be found. In \cite{6} and \cite{7} there are improved algorithms for computing the nerve of a covering. 
\\
\begin{theorem}
(\cite{1} \textbf{Proposition 5.2} ) Given finite set $\chi \subset \mathbb{R}^{n+1}$ and a positive real number $\epsilon$, one can compute the homology of $\cup_{x \in \chi}B(x,\epsilon)$ with $\chi^{\mathcal{O}(n)}$ operations.
\end{theorem}

Now we can give a more adaptive homology group calculation algorithm and give a conjectured complexity.
\\
\begin{theorem}\label{Theorem 3.28} 
If the above conjecture is true and for any k, every two balls $B_{1}, B_{2} \in \mathcal{B}_{k}$ satisfies $d(c_{B_{1}},c_{B_{2}}) \geq \dfrac{1}{2^{k}}$, we can get that the complexity of the algorithm is about $(nsD^{-s})^{\mathcal{O}(n^{2})}\left(\dfrac{1}{A_{n}}\int_{\mathbb{S}^{n}} \left ( \kappa \left (  F,G,x \right ) \right )^{n}dV\right)^{\mathcal{O}(n)}$
\end{theorem}

\begin{algorithm}  
        \caption{Homology(F,G) }  
        \begin{algorithmic}[1] 
            \Require   A semialgebraic system $(F,G) \in \mathcal{H}_d[q + s]$ with $q \leq n$ 
            \Ensure   The homology groups of the set $\{f_1 = ...f_q = 0\}$ and ${g_1 \succ 0,...,g_q \succ 0}\subset R^n$ 
            \State $\mathcal{B}\gets Covering(F,G)$
            \State $\mathcal{N} \gets the nerve of B  $
            \State \Return  the homology group of $\mathcal{N}$

        \end{algorithmic}  
    \end{algorithm}  

\begin{proof}
By \textbf{theorem 3.25}, the cost of computing $\mathcal{B}$ is bounded by $\dfrac{\left(nsD^{-s}\right)^{\mathcal{O}(x)}}{A_{n}\log_{e}2} \int_{\mathbb{S}^{n}} \left ( \kappa \left (  F,G,x \right ) \right )^{n}dV$. In particular, $\sharp\chi \leq \dfrac{\left(16\left(2+A\right)\right)^{n}}{A_{n}\log_{e}2}\int_{\mathbb{S}^{n}}\left(\frac{1}{f}\right)^{n}dV$ where $f = \dfrac{1}{360B_{s}D^{\dfrac{5}{2}}\kappa^{2}(F,G,x)}$ and $A = \dfrac{1}{180B_{s}D^{\dfrac{3}{2}}}$ from \textbf{Theorem 2.6}.  By \textbf{conjecture 3.27}, the complexity of computing the nerve and homology group is bounded by $\mathcal{B}^{\mathcal{O}(n)}$. As a result, the total complexity is bounded by $(nsD^{-s})^{\mathcal{O}(n^{2})}\left(\dfrac{1}{A_{n}}\int_{\mathbb{S}^{n}} \left ( \kappa \left (  F,G,x \right ) \right )^{n}dV\right)^{\mathcal{O}(n)}$
\end{proof}

\end{section}

%TC:endignore

% Word count
\verbatiminput{\jobname.wordcount.tex}

\end{document}